\documentclass{IEEEtran} 

\usepackage{kspcustom}

\title{Benefits of V2V Communication for Autonomous and Connected Vehicles}
\author{Swaroop Darbha, Shyamprasad Konduri, Prabhakar R. Pagilla
\thanks{$^{*}$The authors are with the Department of Mechanical Engineering, Texas A \& M University,
        College Station, TX 77843, USA.
        {\tt\small (dswaroop,konduri,ppagilla)@tamu.edu}}
}

\begin{document}
\maketitle

\begin{abstract}
    In this paper, we investigate the benefits of Vehicle-to-Vehicle (V2V) communication for autonomous vehicles and provide results on how V2V information helps reduce employable time headway in the presence of parasitic lags. For a string of vehicles adopting a Constant Time Headway Policy (CTHP) and availing the on-board information of predecessor's vehicle position and velocity, the minimum employable time headway ($h_{\min}$) must be lower bounded by $2\tau_0$ for string stability, where $\tau_0$ is the maximum parasitic actuation lag. In this paper, we quantify the benefits of using V2V communication in terms of a reduction in the employable time headway: (1) If the position and velocity information of $r$ immediately preceding vehicles is used, then $h_{\min}$ can be reduced to ${4\tau_0}/{(1+r)}$; (2) furthermore, if the acceleration of `$r$' immediately preceding vehicles is used, then $h_{\min}$ can be reduced to ${2\tau_0}/{(1+r)}$; and (3) if the position, velocity and acceleration of the immediate and the $r$-th predecessors are used, then $h_{\min} \ge {2\tau_0}/{(1+r)}$. Note that cases (2) and (3) provide the same lower bound on the minimum employable time headway; however, case (3) requires much less communicated information. 
\end{abstract}

\section{Introduction}

Recent technological advances in V2V and Infrastructure-to-Vehicle (I2V) communications leads to the following question in the context of autonomous vehicles (AVs): What traffic safety benefits and congestion relief can be guaranteed through the use of V2V and/or I2V communication? 

The Constant Time Headway Policy (CTHP) \cite{ioannou1993autonomous} is a commonly employed spacing policy in AVs such as those equipped with Adaptive Cruise Control (ACC) systems, wherein the desired following distance is proportional to the speed of the vehicle, the proportionality constant is referred to as the time headway ($h_w$). ACC systems typically rely only on  on-board information; however, employable time headway is lower bounded by $2 \tau_0$, where $\tau_0$ is the maximum value of parasitic actuation lag \cite{swaroop2001cthp}.

A reduction in employable time headway ($h_w$) leads to higher traffic throughput and in the case of truck platooning, can result in better fuel efficiency via drafting. Typical values of time headway considered in truck platooning are in the range of $0.5-1$ s, which at a speed of 65 mph ($\approx 30$ m/s) equate to a physical spacing of $15-30$ m. In truck platooning if $\tau_0$ is typically around $\tau=0.5$ s, the time headway must be at least $1$ s or $30$ m of physical spacing, which not only decreases the highway capacity but also does not provide any noticeable improvements in fuel efficiency \cite{nowakowski2016operational}. 

While AVs employing CTHP may not require communicated information for string stability, the use of V2V and I2V communications when employing CTHP may seem paradoxical and this point has not clearly been articulated in \cite{Llaster2016ASNS}. The necessity for using communication when employing CTHP in trucks was recently brought to the attention of the authors by Ploeg\cite{ploeg2011cacc} who noticed string instability without V2V or I2V communication in a platoon of three trucks when employing a time headway of about 0.35 seconds. Other numerical simulations and experiments also seem to suggest that employing the acceleration or velocity information of immediately preceding vehicle(s) can help reduce the time headway \cite{nowakowski2016operational, rajamani2002saacc, naus2010cacc, orosz2013accfdbk}. We must point out that the topic of communication and its benefit in reducing time headway was first considered for Semi-Autonomous Cooperative ACC in \cite{rajamani2002saacc}; the points of departure of this work from \cite{rajamani2002saacc} are many fold: (1) in this work, we consider architectures involving multiple predecessors, and (2) we do not feed back acceleration of controlled vehicle.  


In order to reduce the employable time headway, modern systems such as the Cooperative Adaptive Cruise Control (CACC) systems utilize vehicular communication to acquire additional information \cite{naus2010cacc}.  In this work, we quantify the benefits of V2V communication in terms of reduction in the employable time headway in the presence of parasitic lags while guaranteeing string stability. For this purpose, we consider a singular perturbation to model parasitic actuation dynamics. 

The advantage of using V2V communication in improving safety of the vehicles in the platoon has been explored in the literature. For example, the use and benefits of V2V for collision avoidance via emergency lane changing in AVs was discussed in \cite{SwaroopYoon1999}. In an emergency braking scenario, inter-vehicular communication aids coordination among vehicles in a platoon leading to a reduction in the probability of a collision, expected number and severity of collisions \cite{DarbhaChoiITS2010}. Similar results can also be found in \cite{tian2016collisions,chakravarthy2009pileup,tak2016tpcc}, where it was reported that  information from preceding vehicles can be used in reducing collisions and pileups on the highways. Some studies have investigated the effect of the limitation of communication on platoon string stability. For example, in \cite{oncu2014cacc}, some design guidelines for  selecting control and network specifications were presented by considering the effect of delays and sampling in communication channels. Longitudinal spacing controllers that use information obtained via communication were also studied in the literature. For instance,  platoons employing a Constant Spacing Policy (CSP) controller and utilizing vehicular communication were shown to be string stable in \cite{swaroop1994phd}.  Recently, use of CACC for doubling the throughput at traffic intersections via platooning was shown possible in \cite{LPTV-TRC2017, AFKV-2017}.

Since it is generally accepted that the CSP will result in higher capacity and lower fuel consumption, one may reflect on how V2V communication will influence a platoon employing the constant spacing policy. String stability with CSP has only been guaranteed when lead vehicle information is communicated to all other vehicles in the platoon; this is burdensome from the viewpoint of communication, especially when the length of the platoon is large. Moreover, string instability will always occur when every vehicle has information from a finite number of vehicle ahead \cite{ifac-2017konduri}. The situation does not improve even when every vehicle has information of finite number of vehicles in its vicinity (both forward and backward); see \cite{tac-2006yadlapalli, ijes-2010darbha}. For this reason, we have only investigated the benefits of V2V communication on improving time headway using the CTHP.

In this paper, we show the following benefits of V2V communication: (1) If the position, velocity and acceleration information of $r$ immediate preceding vehicles is used, then $h_{\min} = {2\tau_0}/{(1+r)}$; furthermore, $h_{\min} = {4\tau_0}/{(1+r)}$ if only the position and velocity of the `$r$' immediately preceding vehicles is used; (2) if the information from the immediate and the $r$-th predecessors are used then $h_{\min} = {2\tau_0}/{(1+r)}$; (3) furthermore, if only the immediate predecessor information is used, then $h_{\min} = {\tau_0}$. 

The rest of the paper is organized as follows. A brief background and problem description are given in Section~\ref{sec:pb}. Section~\ref{sec:cthp} provides the main results pertaining to improvements in the minimum allowable time headway when using a CTHP controller and for the following scenarios of communicated information (position, velocity, acceleration): (1) `$r$' predecessor vehicles; and (2) immediate predecessor and the $r^{th}$ predecessor. Since communicated information from multiple vehicles may overload the network, the benefit of communicated acceleration information just from the immediate predecessor vehicle is also provided in Section~\ref{sec:cthp}. Numerical simulations along with representative sample of the results are discussed in Section~\ref{sec:sim}. Conclusions are presented in Section~\ref{sec:conc}.

\section{Background and Problem Description}\label{sec:pb}
Consider a platoon of homogeneous vehicles where each vehicle is represented as a point mass. In the presence of parasitic lags, the vehicle model may be described as:
\begin{eqnarray}
	\ddx_i &=& a_i, \nonumber \\
	\tau  \dot a_i + a_i &=& u_i.
	\label{eqn:lagdyn}
\end{eqnarray}
where $x_i(t), a_i(t)$  and $u_i(t)$,  respectively, are the position, acceleration and control input of the $i^{th}$ vehicle at time $t$. The parasitic lag $\tau$ is {\it usually unknown} but may be bounded, i.e., $\tau \in [0, \tau_0]$ where $\tau = 0$ corresponds to instantaneous actuation and $\tau_0$ is the maximum possible value of parasitic lag. From the viewpoint of robustness of string stability, any vehicle following law must guarantee string stability for every value of parasitic lag lying between $0$ and $\tau_0$; we refer to this as {\em robust string stability}. 

Equations in \eqref{eqn:lagdyn} represent a simple linear model of a string of homogeneous vehicles where each vehicle is represented as a point mass. This model is extensively used in vehicle control and is reasonable for the following reasons: (i) Feedback linearization is typically employed for lower level control design rendering the model to be linear (I/O linearized) and homogeneous; (ii) most vehicle maneuvers do not require braking or acceleration inputs to attain their limits; and (iii) past experience using this model and with platooning experiments has been satisfactory; for example, platooning experiments at California PATH have been based on these models.

To evaluate the performance of a spacing policy, we define the spacing errors in the following. Let $d$ be the minimum spacing or standstill distance between every pair of successive vehicles in the platoon. Then, the spacing error for the $i$-th vehicle in the case of CSP is defined as 
\[ e_i := x_i - x_{i-1} + d, \] 
and in the case of CTHP it is defined as 
\[ e_i := x_i - x_{i-1} + d+h_w v_i \]
where $h_w$ is the time headway. Since we are considering only CTHP, we will use the latter definition. Consider the following simple CTHP controller \cite{ioannou1993autonomous,swaroop1994phd}, that is employed in ACC:
\begin{equation}
    u_i = -k_v(v_i-v_{i-1})-k_p(x_i-x_{i-1}+d+h_w v_i),
    \label{eqn:cthp}
\end{equation}
where $k_v,k_p$ are positive gains. Notice that the above controller is completely based on the information obtained using the vehicle on-board sensors. The typical platoon setup using such a controller is shown in Fig.~\ref{fig:cthpctrl}. 
\begin{figure}
    \centering
    \includegraphics[width=0.5\textwidth]{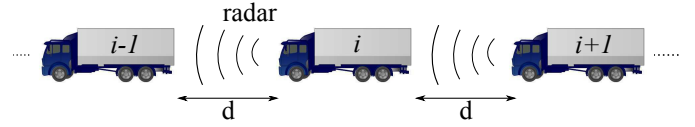}
    \caption{Vehicles equipped with ACC in a platoon where only on-board sensors are used to measure their predecessors' state.}
    \label{fig:cthpctrl}
\end{figure}
Using this controller with the dynamic model \eqref{eqn:lagdyn}, the governing equation for the $i^{th}$ vehicle spacing error is given by
\begin{multline*} 
    \tau \ddde_i + \dde_i+ (k_{v} + k_{p}h_w)\de_i+k_{p}e_i 
    =k_{v}\de_{i-1}+k_{p}e_{i-1}.
\end{multline*}
Let $E_i(s)$ be the Laplace transform of $e_i(t)$. Taking the Laplace transform of the above equation results in the transfer function describing the error propagation (i.e., how the spacing error of the $i^{th}$ vehicle is affected by the spacing error of its predecessor):
\begin{align}
    E_i(s) &= \frac{k_v s+k_p}{\tau s^3+s^2+(k_v+k_ph_w)s+k_p}E_{i-1}(s), \label{eqn:hs} \\
            &:= H(s) E_{i-1}(s), \nonumber
\end{align}
where $H(s)$ is referred to as the spacing error transfer function. 

When every vehicle has information from `$r$' immediate predecessor vehicles, let the control input be
 \begin{multline*}
     u_i = \sum_{l=1}^r [ - k_{vl}(v_i - v_{i-l}) 
                -k_{pl}(x_i - x_{i-l} + d_l + l h_w v_i)],
 \end{multline*}
where the gains $k_{vl}, k_{pl}$ are associated with feeding back velocity and position information of the $l^{th}$ predecessor and $d_l$ is the standstill distance between the $i^{th}$ vehicle and its $l^{th}$ predecessor. The error propagation is governed by 
\begin{equation} 
E_i(s) = \sum_{l=1}^rH_{l}(s)E_{i-l}(s), 
    \label{eqn:rhs}
\end{equation}
where the $H_l(s)$ is given by
\begin{equation*}
    H_{l}(s) = \frac{k_{vl}s+k_{pl}}{\tau s^3+s^2+ \sum_{l=1}^r [\left(k_{vl} + lk_{pl}h_w\right)s+k_{pl}]}.
\end{equation*} 

For such a general error propagation equation, we consider the spacing errors to be states of a spatially discrete system and associate the following characteristic polynomial for the spacing error dynamics when the lead vehicle performs a sinusoidal acceleration maneuver at a frequency $\omega$:
\begin{equation}
    P(z) = z^{r} - \sum_{l=1}^{r}H_l(j\omega)z^{r-l}.
\end{equation}
Let $\rho(P(z; \omega))$ be the spectral radius of $P(z; \omega)$ for a given $\omega$. For string stability, we require that 
\begin{eqnarray}
\label{eqn:ssdef}
\sup_{\omega} \rho(P(z; \omega)) \le 1.
\end{eqnarray} 
In the case of $r=1$, this requirement translates to  $\|H_1(j\omega)\|_{\infty} = \|H(j\omega)\|_{\infty} \le 1$, a frequently imposed frequency-domain condition for string stability. 

A {\em sufficient} condition employed in this paper for ensuring $\rho(P(z; \omega)) \le 1$  is
\begin{equation}
\label{eqn:rssdef}
    \sum_{l=1}^r \|H_l(j\omega) \|_\infty \le 1.
\end{equation}

Define minimum employable time headway ($h_{\min}$) as the minimum time headway for which the platoon is string stable for every $\tau \in [0, \tau_0]$. From \eqref{eqn:hs}, one can show that \cite{swaroop2001cthp}
\begin{equation}
    h_{\min} = 2 \tau_0.
    \label{eqn:hwbasic}
\end{equation} 

In the next section, we show that $h_{\min}$ can be further reduced when vehicular communication is employed. We consider predecessor follower types of information exchange where the communication is unidirectional and the information only flows upstream (the direction of increasing vehicle index). The information flows considered are immediate predecessor (PF), `$r$' immediate predecessors ($r$PF), and immediate and $r$-th predecessors (\rth PF); these are shown in Figure \ref{fig:infoex}.  For parts (b) and (c) of the figure, $r=3$. Notice that when using $r > 1$, in $r$PF, the vehicles with index $i < r$ will use information only from the available predecessors. For instance in part (b) of the Fig.~\ref{fig:infoex}, vehicle 3 uses information only from vehicles 1 and 2 even though $r=3$. Similarly, for the \rth PF case, the \rth-vehicle information is available only to the vehicles with index $i>r$. 
\begin{figure*}
    \centering
   \includegraphics[width=0.8\textwidth]{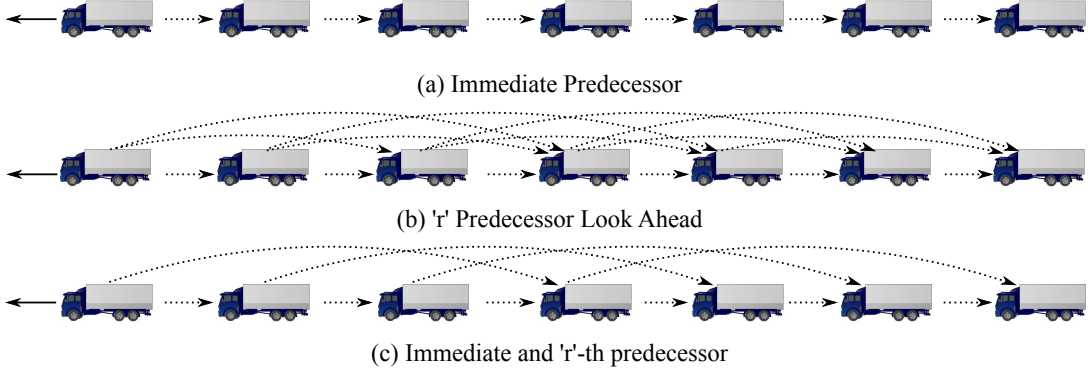}
    \caption{Information flow structures investigated, direction of travel (solid), information flow (dotted). $r=3$ for the structures shown in parts (b) and (c)}
    \label{fig:infoex}
\end{figure*}
\section{CTHP with communicated information}\label{sec:cthp}
We first show that minimum employable time headway can be reduced by using the acceleration of preceding vehicle as in the case of Cooperative ACC. We then show how information from `$r$' immediate predecessors ($r$PF) and the {\rth}PF helps reduce the employable time headway. 

\subsection{CTHP with Immediate Predecessor Acceleration}
Consider a simple control law as follows,
\begin{equation}
u_i = k_aa_{i-1}-k_v(v_i-v_{i-1})-k_p(x_i-x_{i-1}+d+h_wv_i). \label{eqn:cthp-1}
\end{equation}
The governing equation for the $i^{th}$ vehicle spacing error is,
\begin{multline*}
\tau \ddde_i+\dde_i+(k_v+k_ph_w)\de_i+k_pe_i \\
\qquad = k_a\dde_{i-1}+k_v\de_{i-1}+k_pe_{i-1}.
\end{multline*}
Let 
\begin{equation}
    H_e(s) :=  \frac{k_as^2+k_vs+k_p}{\tau s^3 + s^2 + (k_v + k_ph_w)s + k_p},
    \label{eqn:hpf}
\end{equation}
so that errors in maintaining desired following distance propagate as:
$$E_i(s) = H_e(s)E_{i-1}(s). $$
The following theorem quantifies the benefit of using the acceleration information from just the immediate predecessor.

\begin{theorem}
\label{thm:pf}
\begin{itemize}
\itemsep 0.5em
    \item[(a)] $k_a \ge 0$ and $\|H_e(j\omega; \tau)\|_{\infty} \le 1$ for all $\tau \in [0, \tau_0]$ implies $k_a \in (0,1)$ and $h_w \ge \frac{2 \tau_0}{1+k_a}$.
    \item[(b)] Given any $k_a \in (0,1), \eta >0$ and $h_w \ge \frac{2 \tau_0}{1+k_a}$, there exists $k_p, k_v >0$ such that $\|H_e(j\omega; \tau) \| \le 1$ for all $\tau \in [0, \tau_0]$.
\end{itemize}
\end{theorem}
\begin{remark}\label{rem:thminplainwords} The above theorem can be put in words as follows:
    A platoon with individual vehicle dynamics \eqref{eqn:lagdyn} and the control law  \eqref{eqn:cthp-1} can be made robustly string stable with appropriate choice of $k_v, k_p$ iff $k_a\in (0,1)$ and 
    \begin{equation}
        h_w \ge h_{\min} = \frac{2 \tau_0}{1+k_a}.
        \label{eqn:hwpf}
    \end{equation}
\end{remark}
The proof of the above theorem is provided in Appendix. The proof utilizes the necessary condition for attenuation of errors \eqref{eqn:ssdef} to lower bound the minimum employable time headway. 

\vspace*{0.1in}
\begin{remark}\label{rem:thm1recast}
Note that $k_a \ge 0$; otherwise, the feedforward part of the control law \eqref{eqn:cthp-1} will instruct the following vehicle to brake in response to its predecessor's acceleration and vice-versa; moreover, the lower bound for time headway is worse than without using predecessor's acceleration information. For these reasons, we let $k_a \ge 0$.
\end{remark}

\begin{remark} \label{rem:acclnfeedback}
One could modify the control law \eqref{eqn:cthp} by feeding back acceleration of the controlled vehicle as in \cite{rajamani2002saacc}:
$$u_i = -\bar {k}_a(a_i-a_{i-1}) - \bar{k}_v(v_i - v_{i-1})-\bar{k}_p(x_i-x_{i-1}+d+h_wv_i). $$
In this case, the closed loop error evolution equation satisfies
$$ \tau \dddot e_{i} + (1+\bar{k}_a) \ddot e_i = \bar{k}_a a_{i-1} -\bar{k}_v(v_i-v_{i-1}) - \bar{k}_p(x_i - x_{i-1}+d+h_wv_i), $$
which can be recast as:
\begin{align*} 
\underbrace{\frac{\tau}{1+\bar{k}_a}}_{\bar{\tau}} \dddot e_i + \ddot e_i =& \underbrace{\frac{\bar{k}_a}{1+\bar{k}_a}}_{k_a} a_{i-1} - \underbrace{\frac{\bar{k}_v}{1+\bar{k}_a}}_{k_v}(v_i - v_{i-1})\\ &- \underbrace{\frac{\bar{k}_p}{1+\bar{k}_a}}_{k_p}(x_i - x_{i-1} + d + h_w v_i). 
\end{align*}
Note that the recast equation corresponds to CTHP without acceleration feedback of the controlled vehicle, but with modified actuation lag and gains. The proposed methodology can therefore be directly used to analyze this situation as well. 

One may apply the results of Theorem 1 to conclude that 
$$h_w \ge \frac{\frac{2 \tau_0}{1+k_a}}{1+ \frac{k_a}{1+k_a}} = \frac{2 \tau_0}{1+2k_a},$$ and $k_a$ can be chosen to be arbitrarily large as in \cite{rajamani2002saacc}; however, in practice, $k_a$ is limited and one cannot maintain an arbitrarily small time headway with just predecessor information -- clearly, in the limiting case of zero time headway when it reduces to the constant spacing policy, it is impossible to maintain string stability with just predecessor's information. The bound on maximum value of $k_a$ can be determined from a more refined model of parasitic actuation/sensing. The first order singular perturbation model of parasitic dynamics is a simple representation; however, if a third order model were to be chosen instead, there is a maximum value of $k_a$ for which acceleration feedback of controlled vehicle leads to instability. Once the maximum acceleration feedback gain, $k_a$ is determined, Theorem 1 produces a lower bound on the minimum employable time headway.
While acceleration feedback was considered in \cite{rajamani2002saacc}, the bound on the time headway as a function of acceleration feedback gain can be inferred from Theorem \ref{thm:pf}.
\end{remark}
Theorem \ref{thm:pf} will be handy when studying other information architectures in the following subsections. 

\subsection{CTHP with Information from `$r$' Predecessors}
We consider the following generalization of the CTHP control law (\ref{eqn:cthp}) with position, velocity and acceleration information from `$r$' predecessor vehicles: 
 \begin{align}
     u_i(t) &= \sum_{l=1}^r \left[k_{al}a_{i-l}(t)-k_{vl}(v_i(t) - v_{i-l}(t)) \right. \nonumber \\
        & \ \ \ \left. -k_{pl}(x_i(t) - x_{i-l}(t) + d_l + l h_w v_i(t))\right], \label{eqn:cthp-general}
 \end{align}
where the gains $k_{al},k_{vl}$, and $k_{pl}$ are associated with feeding back the acceleration, velocity and position information associated with the $l^{th}$ predecessor and $d_l$ is the standstill distance between the $i^{th}$ vehicle and its $l^{th}$ predecessor. The above generalized control requires information that can be obtained only via vehicular communication. We investigate whether it is possible for the platoon to be robustly string stable while reducing the minimum employable time headway. 

Substituting the control law \eqref{eqn:cthp-general} into \eqref{eqn:lagdyn}, the governing equation for the $i^{th}$ vehicle spacing error is given by
\begin{multline*} 
    \tau \ddde_i + \dde_i+ \sum_{l=1}^r[\left(k_{vl} + lk_{pl}h_w\right)\de_i+k_{pl}e_i] \\
    \qquad =\sum_{l=1}^r(k_{al}\dde_{i-l}+k_{vl}\de_{i-l}+k_{pl}e_{i-l}).
\end{multline*}
The propagation of the spacing error to vehicle $i$ from $r$ predecessor vehicles is given by
\[ E_i(s) = \sum_{l=1}^rH_{pl}(s)E_{i-l}(s), \] 
where
\begin{equation}
    H_{pl}(s) = \frac{k_{al}s^2+k_{vl}s+k_{pl}}{\tau s^3+s^2+ \sum_{l=1}^r [\left(k_{vl} + lk_{pl}h_w\right)s+k_{pl}]}.
    \label{eqn:rtf}
\end{equation} 
The transfer function $H_{pl}(s)$ describes the effect of the spacing error in the $(i-l)^{th}$ vehicle on the spacing error of the $i^{th}$ vehicle. 
Since our objective is to demonstrate the benefits of V2X communication with some controller obeying the architecture considered, it suffices to choose the same set of gains so that we may use the result in Remark~\ref{rem:thm1recast}. For all $l$, let $k_{al} = k_a, k_{vl} = k_v, k_{pl} = k_p$, so that we may define $$H_0(s) := \frac{k_a s^2+k_vs+k_p}{\tau s^3+s^2+ (rk_v + \frac{r(r+1)}{2}k_ph_w)s + rk_p}, $$  
and the error propagation may be described by:
$$E_i(s) = H_0(s) \sum_{l=1}^r E_{i-l}(s). $$
We may now use the interpretation of Theorem 1 given in Remark 1 and the sufficient condition given by equation \eqref{eqn:rssdef} to get the following result:
\begin{theorem}\label{thm:rv}
A platoon with individual vehicle dynamics \eqref{eqn:lagdyn} and each vehicle receiving information from `$r$' predecessors as given by the control action in \eqref{eqn:cthp-general}, where $k_{al} = k_a$ and $rk_a \in (0,1)$, $k_{vl} = k_v$, $k_{pl}=k_p$, is robustly string stable when
    \begin{equation}
        h_{\min} = \frac{4 \tau_0}{(1+r)(1+rk_a)}.
        \label{eqn:hwrpf}
    \end{equation}
\end{theorem}
\begin{proof}
Consider
\[ rH_0(s)= \frac{rk_as^2 + rk_vs + rk_p}{\tau s^3+ s^2+ (rk_v + \frac{r(r+1)}{2}k_p h_w) s + r k_p}. \]
Define $\bar k_a:= r k_a, \; \bar k_v = r k_v, \; \bar k_p = r k_p, \; \bar h_w = \frac{r+1}{2} h_w$, and $\bar H_0(s) = r H_{0}(s)$, then
\[ \bar H_0(s) = \frac{\bar k_as^2+ \bar k_v s+ \bar k_p}{\tau s^3+ s^2 + (\bar k_v + \bar h_w \bar k_p)s +\bar k_p}. \]
Comparing $H_e(s)$ from Theorem~\ref{thm:pf} with  $\bar H_{0}(s)$, we can conclude that for every ${\bar k_a} \in [0,1)$, there exist a set of gains $({\bar k}_v, {\bar k}_p)$ such that $\|r H_0(j \omega)\|_{\infty} \le 1$ for every $\tau \in [0, \tau_0]$ if 
\begin{eqnarray*}
{\bar h}_w &\ge& \frac{2 \tau_0}{1+{\bar k}_a}, \\
\Rightarrow 
h_w &\ge& \frac{2}{(1+r)} \frac{2\tau_0}{(1+rk_a)} 
= \frac{4 \tau_0}{(1+r)(1+rk_a)}. 
\end{eqnarray*}
\end{proof}
\subsection{CTHP with Immediate and \rth Predecessor Information}
Practical considerations on the communication bandwidth may force each vehicle to pick only a few predecessors to maintain reliable communication; in such situations, one may want to use the immediate predecessor and a second predecessor (\rth \ vehicle) from the downstream of the platoon. When using immediate and \rth-predecessor information in the feedback, the control law can be rewritten as,
 \begin{align}
     u_i(t) &= \sum_{l=1,r} \left[k_{al}a_{i-l}(t)-k_{vl}(v_i(t) - v_{i-l}(t)) \right. \nonumber \\
        & \ \ \ \left. +k_{pl}(x_i(t) - x_{i-l}(t) + d_l + l h_w v_i(t))\right]. \label{eqn:cthp-rth}
 \end{align}
The above control law is a special case of the $r$ vehicle look ahead control law in ~\eqref{eqn:cthp-general} with two vehicle feedback, where the second vehicle is the $r^{th}$  vehicle. 
If $k_{a1} = k_{ar} = k_a \in (0, \frac{1}{2}), \; \; k_{v1} = k_{vr} = k_v,$ and $k_{p1} = k_{pr} = k_p$, define 
$$H_0(s) := \frac{k_a s^2+k_vs+k_p}{\tau s^3+s^2+ (2k_v + 3k_ph_w)s + 2k_p}, $$  
so that the error propagation is given by:
$$E_i(s) = H_0(s)E_{i-1}(s)+H_0(s)E_{i-r}(s). $$
From the sufficient condition given by equation \eqref{eqn:rssdef}, robust string stability can be guaranteed if $2 \|H_0(j \omega)\|_{\infty} \le 1$. 
The proof of the following result is analogous to the proof of  Theorem~\ref{thm:rv}.
\begin{corollary}\label{cor:rthpf}
A platoon with individual vehicle dynamics \eqref{eqn:lagdyn} and the control law  \eqref{eqn:cthp-rth}, where  $k_{a1} = k_{ar} = k_a \in (0, \frac{1}{2}), \; \; k_{v1} = k_{vr} = k_v, \; k_{p1} = k_{pr} = k_p$, is robust string stable when 
\begin{equation}
    h_{\min} = \frac{4 \tau_0}{(1+r)(1+2k_a)}.
    \label{eqn:rthbound}
\end{equation}
\end{corollary}

The following observations are made based on the results of Theorems~\ref{thm:pf}, \ref{thm:rv} and its corollary. These observations may be helpful in making design choices for implementing CACC systems with information from multiple vehicles in the feedback.
\begin{enumerate}
    \item[(I)] When information from only immediate predecessor is used, and acceleration feedback gain is selected to be zero, the lower bound in inequality \eqref{eqn:hwpf} reduces to $h_{min}$ given in equation \eqref{eqn:hwbasic}.
    \item[(II)] In the immediate predecessor feedback case if the acceleration feedback gain $k_a$ is chosen arbitrarily close to one, the minimum employable time headway can be chosen close to $\tau_0$; this can be inferred from equation \eqref{eqn:hwpf}. Hence, choosing $k_a \approx 1$, the lower limit on $h_w$ can be nearly halved and the platoon can be string stable for any headway greater than the maximum parasitic lag.
    %
    \item[(III)] If information from two predecessor vehicles ($r=2$) is utilized with equal gains selected for position, velocity and acceleration feedback, and if the acceleration feedback gains are chosen such that their sum is close to unity, then the minimum employable time headway can be chosen close to $\frac{2\tau_0}{3}$; this is also corroborated in \cite{orosz2014accfdbk} via numerical simulations. 
    \item[(IV)] For $r = 3$, when equal gains are selected for position and velocity feedback, there are two possible scenarios: (1) when $k_a=0$ for every vehicle and (2) when $rk_a$ is close to one. For the former, the minimum time headway can be chosen close to $\tau_0$. For the latter, it can be chosen close to $\tau_0/2$. Thus, every vehicle only needs to use information from at least three predecessor vehicles in order to overcome the limitation imposed by parasitic lag in a vehicle. This is intuitive because if the vehicle has access to information of vehicles downstream, then it has knowledge of future events to come, which is not possible with information from only the immediate predecessor vehicle. 
    \item[(V)] In order to overcome the handicap of not having acceleration information of preceding vehicles, one must have velocity and position information of more preceding vehicles. For example, comparing the case $r=2$ with acceleration feedback as considered in (III) with the case $r=3$ without acceleration feedback as considered in case (1) of (IV), the importance of acceleration feedback becomes clear. To match the reduction in time headway obtained in (III), every vehicle must have the velocity and position information of five predecessor vehicles.  
    %
    %
    %
\end{enumerate}
\begin{table}
\centering
\begin{tabular}{|l|c|}
	\hline
	Communication Type & $h_{\min}$ \\ \hline
	Immediate predecessor & $2\tau_0/(1+k_a)$\\ \hline
	$r$ predecessors & ${4\tau_0}/{(1+r)(1+rk_a)}$ \\ \hline
	With $rk_a \approx 1$  & ${2\tau_0}/{(1+r)}$ \\ \hline
	With $rk_a = 0$  & ${4\tau_0}/{(1+r)}$ \\ \hline
	With $rk_a \approx 1,r=2$ & ${2\tau_0}/{3}$ \\ \hline
	With $rk_a \approx 1,r=3$ & ${\tau_0}/{2}$ \\ \hline
	Immediate and $r^{th}$ predecessor & ${4\tau_0}/{(1+r)(1+2k_a)}$ \\ \hline
	With $2k_a \approx 1$ & ${2\tau_0}/{(1+r)}$ \\ \hline
\end{tabular}
\caption{Communication type vs the minimum employable time headway.}
\label{tab:benefits}
\end{table}
The communication models explored and their respective minimum employable time headway are summarized in Table~\ref{tab:benefits}.
While theoretically using either large number of vehicles or the \rth \ predecessor in the feedback improves the capacity of the highway by reducing the lower bound on the minimum employable time headway, this may cause additional communication overhead. Furthermore, the current state of the art uses near field communication technology \cite{nowakowski2016operational} which imposes severe restriction on the value of $r$ that can be used. In light of this, a good compromise will be to use information from two or three predecessor vehicles that are close to the vehicle. For example, one can use information from the two immediate predecessors or the immediate and the $3^{rd}$ predecessor. In the cases where communication bandwidth is restricted, using only information from the immediate predecessor will also provide considerable benefits over using information available from just the on-board sensors.

\subsection{Non-negativity of impulse response for string stability}
The error propagation equation described by \eqref{eqn:rhs} is constrained so that $\sum_{l=1}^r H_l(0) = 1$. When $r=1$, this constraint implies $H(0) = 1$ in equation \eqref{eqn:hs}.
Let $h(t)$ denote the impulse response of $H(s)$. An alternate definition of  string stability stems from the requirement that spacing errors must not amplify as we move upstream (increasing index) along the platoon. For $r=1$, we readily have:
\[
\|e_{i}\|_{\infty} \leq  \|h\|_1 \|e_{i-1}\|_{\infty},
\]
and hence, $\|h\|_1 \leq 1$ is also used as a criterion for string stability. As,
\[ 1= |H(0)| \le \|H\|_\infty \le \|h\|_1. \]
If $h(t) \ge 0$, then
\[ 1= |H(0)| = \|H\|_\infty = \|h\|_1. \]
Hence, the condition $\|h\|_1 \leq 1$ is equivalent to the frequency domain condition  
\begin{equation}
    \|H(j\omega)\|_\infty \le 1, \ \mbox{when} \ h(t) \ge 0, \ \forall t.
    \label{eqn:ssdefimp}
\end{equation}
For the multiple vehicle look-ahead case, from equation \eqref{eqn:rhs}, one can show from I/O properties of linear systems that:
\[ \|e_i\|_{\infty} \le \sum_{l=1}^r \|h_l(t)\|_1 \|e_{i-l}(t)\|_{\infty}, \]
where $h_l(t)$ is the unit impulse response of $H_l(s)$.  In this case, a sufficient condition for errors not to amplify geometrically (at least asymptotically) is that $\sum_{l=1}^r \|h_l\|_1 \le 1$; together with the constraint that $\sum_{l=1}^r H_l(0) =1$ which can only be satisfied if $h_l(t) \ge 0$ for every $l$. Hence, string stability of the platoon can be investigated by studying the peak magnitude of $H(s)$ and imposing an additional requirement on the non-negativity of the impulse response. 

Even for the case of $r=1$, the additional requirement of non-negativity of impulse response renders the problem difficult, as one must prove the following counterpart of Remark 1: Given $k_a \in (0,1)$, there exist $k_p, k_v$ such that $h(t) \ge 0$ for every $\tau \in [0, \tau_0]$ whenever $h_w \ge \frac{2\tau_0}{1+k_a}$. This seemingly simple problem is analytically difficult to solve and is related to the open problem of finding a fixed structure controller satisfying a transient specification (namely, the impulse response of the transfer function is non-negative).  

A transformation involving scaling with respect to $\tau_0$ of the above problem leads it to a standard form (involving one less variable) where proving the following result suffices: Given $k_a \in (0,1)$ and $h_w \ge \frac{2}{1+k_a}$, there exist $k_p, k_v$ such that $h(t) \ge 0$ for every $\tau \in [0, 1]$. Later, we provide a set of gains $k_p, k_v$ for a given $k_a$ that results in $h(t)$ being non-negative. The purpose of showing non-negativity (at least numerically) is to show that even if one were to choose another criterion for string stability, the results presented in this paper will continue to hold.

The basic idea of the transformation is as follows: Let $s = s'/\tau_0$, then the error propagation transfer function becomes,
\[ H_e(s'/\tau_0) = \frac{k_as'^2/\tau_0^2+k_vs'/\tau_0+k_p}{\tau s'^3/\tau_0^3 + s'^2/\tau_0^2 + (k_v + k_ph_w)s'/\tau_0 + k_p}.\]
Multiplying both the numerator and the denominator with $\tau_0^2$ results in,
\[ H_e(s'/\tau_0) = \frac{k_as'^2+k_vs'\tau_0+k_p\tau_0^2}{\tau s'^3/\tau_0 + s'^2 + (k_v + k_ph_w)s'\tau_0 + k_p\tau_0^2}.\]
Let $\tilde \tau = \tau/\tau_0$, $\tilde k_p = k_p\tau_0^2$, $\tilde k_v = k_v\tau_0$, $\tilde h_w = h_w/ \tau_0$ then

\begin{equation} 
    \tilde H_e(s') := H_e(s'\tau_0) = \frac{k_as'^2+\tilde k_vs'+\tilde k_p}{\tilde \tau s'^3 + s'^2 + (\tilde k_v + \tilde k_p \tilde h_w)s' + \tilde k_p},
    \label{eqn:tildehe}
\end{equation}
where $\tilde \tau \in [0,1]$ and $\tilde h_w = 2/(1+k_a)$. The above representation considerably simplifies the analysis and it suffices to show that there exist gains $\tilde k_v$, $\tilde k_p$ such that $\tilde h_e(t) \ge 0$. 

Since $\tilde H_e(s')$ is of the same form as $H(s)$ in Theorem 1, we can relabel all the variables ($\tilde k_v$ to be $k_v$, $\tilde k_p$ to be $k_p$, $\tau_0= 1$ etc) and use them interchangeably. In the following, we will focus on determining $k_p, k_v$ for a given $k_a$ when $h_w \ge \frac{2 \tau_0}{1+k_a}$. In this connection, we will use two results available in the literature:

\begin{enumerate}
\item When $\tau = 0$, let $-z_1$, $-z_2$ and $-p_1$, $-p_2$ denote the real and distinct location of zeros and poles of the above transfer function, respectively. Further let, $z_1 < z_2$ and $p_1 < p_2$. The impulse response is non-negative if \cite{swaroop1994comparision},
\begin{equation}
    p_1 \le \frac{\tilde h_w \tilde k_p}{1-k_a} \le p_1+p_2.
    \label{eqn:impt0}
\end{equation}
\item When $\tilde \tau \in (0,1]$, let the three real and distinct poles of the transfer function be located at $-p_1$, $-p_2$, and $-p_3$, and let $p_3>p_2>p_1$. Then, the impulse response of the system is,
\[ \tilde h_e(t) = c_1e^{-p_1t}+c_2e^{-p_2t}+c_3e^{-p_3t},\]
where $c_1,c_2$ and $c_3$ are the residues obtained from the partial fraction expansion of the transfer function $\tilde H_e(s')$. The impulse response is non-negative if the residues satisfy \cite{lin1997nnir}:
\begin{equation}
    c_1 \ge 0, c_2 < 0 \mbox{ and } c_3 > \frac{p_2-p_1}{p3-p1}c_2.
    \label{eqn:impt}
\end{equation}
Thus, for a given $k_a$, $\tau_0$ and $\tilde h_w = 2/(1+k_a)$, any set of $\{\tilde k_p, \tilde k_v\}$ that satisfy the relations in \eqref{eqn:impt0} and \eqref{eqn:impt} will guarantee $\tilde h_e(t) \ge 0$. 
\end{enumerate}
In the case of $k_a = 0.95$,  $\tau_0 =1$, the following gains seem to indicate numerically that $h_e (t) \ge 0$:
\[ \{k_v,k_p\} = \{0.082, 0.001\}.\]
Furthermore, these gains can be scaled back for any other $\tau_0$ using $\tilde k_p = k_p \tau_0^2$ and $\tilde k_v = k_v \tau_0$ to guarantee non-negative impulse response of $H_e(s)$. 

From the construction of the proof of Theorem 2 and its corollary, it is clear that demonstrating the benefits of V2X communication with multiple vehicle look-ahead even with the additional non-negativity requirement will reduce to showing the benefits for the single vehicle look-ahead information, which was the focus of the discussion in this section. 

\section{Numerical Simulations}\label{sec:sim}
In this section, we discuss numerical simulations that corroborate the results of Theorem~\ref{thm:rv} and its corollaries. The CTHP controller from \eqref{eqn:cthp-general} is considered for the simulations.
The numerical values for the common parameters are given in Table~\ref{tab:cthp-values}. A sinusoidal disturbance is applied on the lead vehicle between 5 and 10 seconds of the simulation time. Only the odd numbered vehicles are shown in the figures to reduce clutter in the plots. 
\begin{table}[h!]
\centering
\begin{tabular}{ |c|c|c|c|c|c|c|} 
 \hline
 $n$ & $d$ & $\tau_0$ & $k_p$ & $k_v$ & $k_a$ & $v_r$ \\\hline 
 15 & 5 m & 0.5 & 45 & 0.8 & 0.25 & 20 m/s \\ 
 \hline
\end{tabular}
\caption{Numerical values}
\label{tab:cthp-values}
\end{table}
Two time headway cases are considered: (i) $ h_w > h_{\min}$ and (ii) $h_w < h_{min}$. We also considered three values for `$r$': (1) $r=1$, (2) $r=2$, and (3) $r=3$. Figures~\ref{fig:r1} through ~\ref{fig:r3ka} provide the simulation results with the numerical values as given in Table~\ref{tab:cthp-cases}; the column entitled $h_w$ (a) satisfies the lower bound and the column $h_w$ (b) violates the lower bound.
\begin{table}[h!]
\centering
\begin{tabular}{ |c|c|c|c|c|c|c|} 
 \hline
  Figure & $r$ & $k_a$ & $h_{\min}$ & $h_w$ (a) & $h_w$ (b)  \\\hline 
 Fig.~\ref{fig:r1}  & 1 & 0.25 & 0.8 & 0.88 & 0.68 \\ \hline
 Fig.~\ref{fig:r2ka0}  & 2 & 0 & 0.66 & 0.8 & 0.63 \\ \hline
 Fig.~\ref{fig:r2ka}  & 2 & 0.25 & 0.44 & 0.68 & 0.4 \\ \hline
 Fig.~\ref{fig:r3ka0}  & 3 & 0 & 0.5 & 0.6 & 0.47 \\ \hline
 Fig.~\ref{fig:r3ka}  & 3 & 0.25 & 0.28 & 0.5 & 0.27 \\
  \hline
\end{tabular}
\caption{Numerical values corresponding to figures}
\label{tab:cthp-cases}
\end{table}
It is clear from the figures and the numbers in the table above that lower time headway values can be employed when V2V communication is used.

Figure \ref{fig:nnir_ka095} provides position and velocity gain values (with $k_a=0.95$) that correspond to error propagation transfer function subjected to non-negative impulse response requirement. The gain values $\{\tilde k_v, \tilde k_p\}$ were determined using the discussions in Section III.D to ensure $\tilde h_e(t) \ge 0$ when $k_a = 0.95$, $\tau_0 =0.5$ seconds, and $\tilde h_w = 2/(1+k_a) = 1.02$. The approach involved utilization of a range of parameter values (gridding) for $\tilde k_v$ and $\tilde k_p$ and checking if the conditions for non-negative impulse response requirement were met. In Fig.~\ref{fig:nnir_ka095}, the blue dots correspond to $\tilde h_e(t) \ge 0$ when $\tilde\tau=0$; the red dots correspond to real poles for the time-scaled spacing transfer function ($\tilde H_{e}(s)$) and $\tilde\tau \in [0,1]$; the orange circles correspond to the three conditions, $\tilde\tau \in [0,1]$, real poles for $\tilde H_{e}(s)$, and $\tilde h_e(t) \ge 0$. Therefore, in Fig.~\ref{fig:nnir_ka095}, the orange shaded region provides the permissible position and velocity gain values. We have also evaluated the impulse responses for different set of gains and various values of parasitic lags ($\tilde \tau \in [0,1]$) and found them to be all non-negative.  

\begin{figure}[p]
    \centering
     \includegraphics[width=0.45\textwidth]{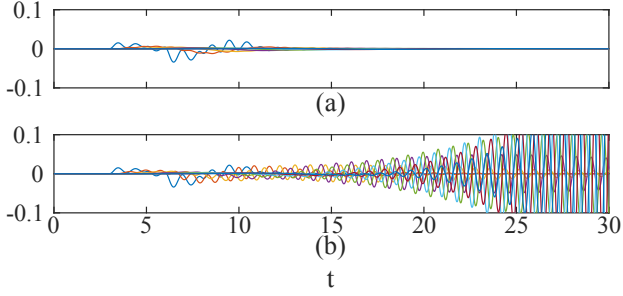}
    \caption{CTHP with $r=1$.}
    \label{fig:r1}
\end{figure}
\begin{figure}
    \centering
     \includegraphics[width=0.45\textwidth]{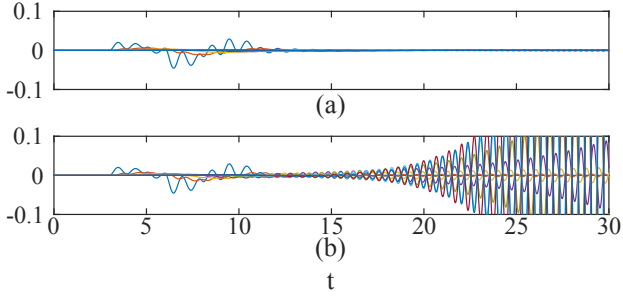}
    \caption{CTHP with $r=2$, $K_a = 0$.}
    \label{fig:r2ka0}
\end{figure}
\begin{figure}
    \centering
     \includegraphics[width=0.45\textwidth]{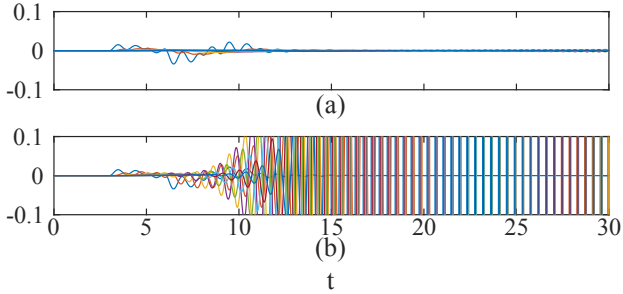}
    \caption{CTHP with $r=2$}
    \label{fig:r2ka}
\end{figure}
\begin{figure}
    \centering
     \includegraphics[width=0.45\textwidth]{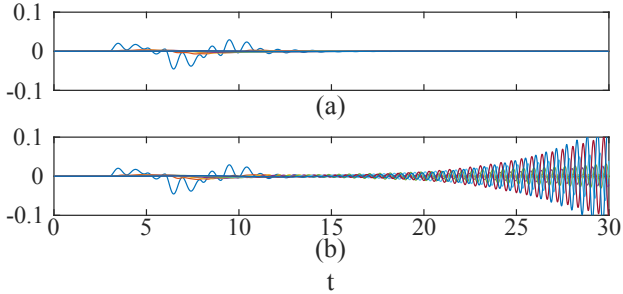}
    \caption{CTHP with $r=3$, $K_a = 0$}
    \label{fig:r3ka0}
\end{figure}
\begin{figure}
    \centering
     \includegraphics[width=0.45\textwidth]{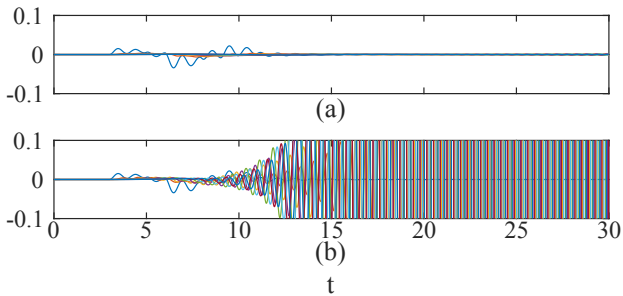}
    \caption{CTHP with $r=3$.}
    \label{fig:r3ka}
\end{figure}
\begin{figure}
    \centering
     \includegraphics[width=0.45\textwidth]{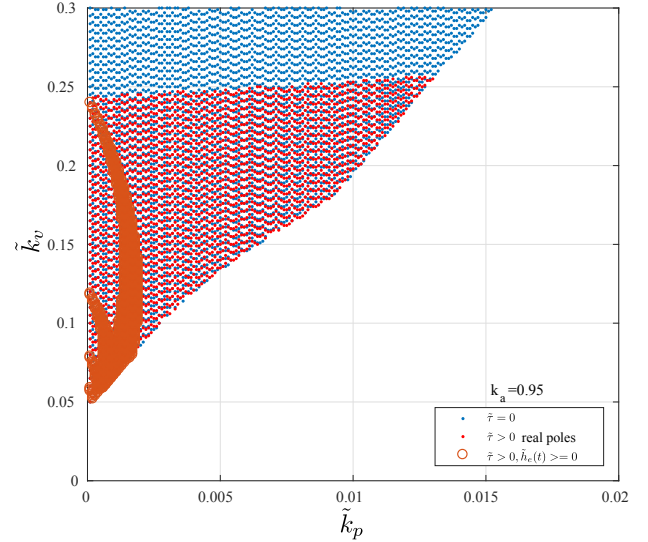}
    \caption{Gains $\tilde k_p,\tilde k_v$ that satisfy $\tilde h_e(t) \ge 0$ when $k_a = 0.95$ and $\tilde h_w = 1.02$.}
    \label{fig:nnir_ka095}
\end{figure}
%

\section{Conclusions}\label{sec:conc}
In this paper, we have studied the benefits of using information obtained via V2X (V2V or V2I) communication on the performance of the autonomous vehicular platoons. We have demonstrated that using a CTHP controller with information from `$r$' predecessor vehicles that the platoon is robustly string stable and further decreases the minimum employable time headway, thereby increasing the capacity of the platoon. Therefore, using V2X communication to feedback information from predecessor vehicles is beneficial provided a proper spacing policy is employed. As discussed in Sections~\ref{sec:cthp} and \ref{sec:sim}, the bound on the time headway may not be tight and there may be a possibility of reducing this further. Hence, a possible future research direction will be to find such a bound.

\bibliography{acc2017-benefits}
\bibliographystyle{IEEEtran}

%
%
\appendix{Proof of Theorem~\ref{thm:pf}}\label{appx:pfproof}

\begin{proof} 
\begin{itemize}
    \item[(a)] Consider the error propagation transfer function
\[ 
    H_e(s) := \frac{N(s)}{D(s)} = \frac{k_as^2+k_vs+k_p}{\tau s^3 + s^2 + (k_v + k_ph_w)s + k_p}.
\]
Let $s= j \omega$, so that
\begin{multline*}
	\|H_e(j\omega; \tau)\|^2 = \frac{(k_p - k_a \omega^2)^2 + k_v^2\omega^2}{(k_p-\omega^2)^2+ \omega^2((k_v + k_p h_w - \tau \omega^2)^2}.
\end{multline*}
Note that $D(s)$ is Hurwitz for every $\tau \in [0, \tau_0]$ if and only if 
\begin{equation}
    k_v+k_ph_w > \tau_0 k_p >0 \iff h_w > \tau_0 -\frac{k_v}{k_p}.
    \label{eqn:kastability}
\end{equation}
This is a standard application of the conditions for a third-order polynomial to be Hurwitz. Corresponding to $h_w=0$, \cite{its-2016zheng} provides a similar condition for a constant spacing policy. The above equation is a basic requirement of stability and it imposes a limit on the allowable time headway; in particular, a smaller value of $k_p$ leads to a smaller lower bound for $h_w$. If $k_v\le 0$, then $h_w \ge \tau_0$. Since we want a tighter bound on $h_w$, we must choose $k_v \ge 0$. 


For showing necessity, we first prove that $k_a \in (0,1)$. Define \[ \omega_0 := \sqrt{{k_v+h_wk_p}/{\tau}}.\] 
If $k_a>1$, consider the frequency frequency $\omega_0$, and any $\tau < \min\{h_w, \tau_0\}$.
Then $\omega_0^2 > k_p > \frac{k_p}{k_a}$.
Hence, $(\omega_0^2 - \frac{k_p}{k_a}) > \omega_0^2 - k_p >0$, and 
\begin{multline*}
	\|H_e(j\omega_0; \tau)\|^2 = k_a^2 \frac{(\frac{k_p}{k_a} - \omega_0^2)^2 + \frac{k_v^2}{k_a^2}\omega_0^2}{(k_p-\omega_0^2)^2}\ge k_a^2 >1.
\end{multline*}
If $k_a = 1$, it is clear that 
$\|H_e(j\omega_0; \tau)\|^2 >1$ if $k_v^2 \ge (k_v+k_ph_w-\tau \omega^2)^2$ for some $\omega$; this would clearly be the case if $\omega^2 \in (\frac{k_ph_w}{\tau}, \frac{k_ph_w + 2k_v}{\tau})$.
Therefore, $k_a \in (0,1)$ and $1-k_a^2 > 0$. 

Now we will show the bound on $h_{\min}$. Consider 
\[ 
    H_e(s) := \frac{N(s)}{D(s)} = \frac{k_as^2+k_vs+k_p}{\tau s^3 + s^2 + (k_v + k_ph_w)s + k_p}.
\]
Let $s=j\omega$, then,
\begin{equation*}
	\|H_e(j\omega; \tau)\|^2=\frac{ (k_p-k_a\omega^2)^2 + k_v^2\omega^2}{(k_p-\omega^2)^2 + \omega^2(k_v+h_wk_p - \tau \omega^2)^2}.
\end{equation*}
Let us define 
\begin{align*}
    \Delta_N(\omega^2; \tau) &:= (k_p-k_a\omega^2)^2 + k_v^2\omega^2, \ \mbox{and} \\ 
    \Delta_D(\omega^2; \tau) &:=(k_p-\omega^2)^2 + \omega^2(k_v+h_wk_p - \tau \omega^2)^2. 
\end{align*}
Then 
\[ \|H_e(j\omega; \tau)\| \le 1 \iff \Delta_N(\omega^2; \tau)- \Delta_D(\omega^2; \tau) \ge 0, \ \forall \omega. \] 
Substituting for $\Delta_N$, $\Delta_D$ and simplifying,
\begin{multline}
    	\tau^2 \omega^4 + \omega^2[(1-k_a^2)-2 \tau (h_w k_p + k_v)]\\ 
    	+ (k_v+h_wk_p)^2-k_v^2-2k_p(1-k_a) \ge 0 .
    	\label{eqn:kaperturbed}
\end{multline} 
The above inequality is a bi-quadratic inequality; for $\tau=0$, this inequality corresponds to  
\begin{equation}
	(k_v+h_wk_p)^2-k_v^2-2k_p(1-k_a) \ge 0. 
	\label{eqn:kanominal}
\end{equation}
When $\tau \neq 0$, the bi-quadratic inequality  \eqref{eqn:kaperturbed} holds for all $\omega \in \Re$ and $\tau \in (0, \tau_0]$ if and only if for every $\tau$, 
if either  (A) the relation \( (1-k_a^2)-2 \tau (h_wk_p+k_v) \ge 0 \) holds, or (B) the discriminant of equation \eqref{eqn:kaperturbed} is non-positive.

For case (A), the relation \( (1-k_a^2)-2 \tau (h_wk_p+k_v) \ge 0 \) together with the nominal case in \eqref{eqn:kanominal} implies that 
	\begin{align*}
		(k_v+h_wk_p)^2  &\ge k_v^2+2k_p(1-k_a), \\
	    \Rightarrow (1+k_a)(k_v+h_wk_p)^2 &- 4\tau k_p (k_v+h_wk_p) \\ 
		& \ge k_v^2(1+k_a) \ge k_v^2.
	\end{align*}
Completing the square and noting that $k_v+h_wk_p>0$, we get
{\small
\begin{align*}
		k_v+h_wk_p &\ge \frac{2k_p\tau}{1+k_a} + \sqrt{\frac{4 \tau^2k_p^2}{(1+k_a)^2}+k_v^2} 
		\ge \frac{2k_p\tau}{1+k_a} + k_v, \nonumber \\
\end{align*}
}
which implies, \[ h_w \ge \frac{2 \tau}{1+k_a}. \]
Since this should be true for every $\tau \in (0, \tau_0]$, it must be true in this case that the inequality \eqref{eqn:hwpf} is true.

For case (B), the discriminant of equation \eqref{eqn:kaperturbed} is non-positive; specifically from \eqref{eqn:kaperturbed}:
\begin{align*}
[(1-k_a^2) &-2\tau(h_wk_p+k_v)]^2 \\
&\le 4\tau^2[(k_v+h_wk_p)^2-k_v^2-2k_p(1-k_a)], \\
\Rightarrow &h_wk_p + k_v \ge \frac{(1-k_a)^2+4 \tau^2k_v^2+8 \tau^2k_p(1-k_a)}{4 \tau (1-k_a^2)}, \\
\Rightarrow &h_wk_p \ge \frac{2 \tau}{1+k_a} k_p + \frac{(2 \tau k_v - (1-k_a)^2)^2}{4 \tau (1-k_a^2)}.
\end{align*}
Since the above inequality holds for every $\tau \in [0, \tau_0]$, 
$$h_w \ge \frac{2 \tau_0}{1+k_a}.$$

\item[(b)]
To demonstrate sufficiency, we still have to show the following: For any given $\eta>0$ and $k_a \in (0,1)$, one can find $k_p, k_v$ such that the following three conditions hold in order for $\|H_e(j\omega; \tau)\| \le 1$ for all $\tau \in (0, \tau_0]$:
\vspace*{0.1in}
\begin{enumerate}
\item {Stability:} Equation \eqref{eqn:kastability};
\item Nominal Case $(\tau = 0)$: Equation  \eqref{eqn:kanominal}, and
\item Perturbed Case $(\tau \ne 0)$: the family of polynomials given by Equation \eqref{eqn:kaperturbed}, $\forall \quad \tau \in (0, \tau_0] $ is non-negative.
\end{enumerate} 
\vspace*{0.1in}

\noindent 1) Stability: If $h_w \ge \frac{2 \tau_0}{1+k_a} (1+ \eta)$, 
then $k_p > 0$ implies
	\begin{align*}
		h_w k_p 
		&= \frac{2 \tau_0}{1+k_a} k_p + \frac{2 \tau_0 \eta}{1+k_a} k_p 
		\ge \tau_0 k_p + \frac{2 \tau_0 \eta}{1+k_a}k_p > \tau_0 k_p >0. 
	\end{align*}
	The last inequality follows from 
	\[ k_a \in (0,1) \Rightarrow 1 < 1+k_a < 2 \Rightarrow \frac{2}{1+k_a} > 1. \]
	Since $k_v >0$, it follows that $k_v + h_wk_p > \tau_0 k_p$. 
	The stability condition is readily satisfied by the choice of $k_v, k_p >0$, so it does not impose further restrictions on $k_v$ and $k_p$. 
\vspace*{0.1in}

\noindent 2) Nominal Case: Upon simplification of equation \eqref{eqn:kanominal}, condition (b) (in the statement of Theorem 1) is equivalent to satisfying the inequality
	\( h_w(2k_v + h_w k_p) \ge 2(1-k_a). \)
	The set of $k_v, k_p$ that satisfy the above inequality when $h_w = \frac{2 \tau_0 (1+\eta)}{1+k_a}$ is given by:
		$${\mathcal S}_1 := \{(k_p, k_v): k_p >0, \; k_v>0, \; \frac{k_v}{a_1} + \frac{k_p}{b_1} \ge 1\}, $$	
	where
		$$a_1 := \frac{(1-k_a^2)}{2 \tau_0(1+\eta)}, \quad b_1 := \frac{(1+k_a)^2(1-k_a)}{2\tau_0^2(1+\eta)^2}. $$

\vspace*{0.1in}
\noindent 3) Perturbed Case: The set of $k_v, k_p$ satisfying the inequality $(1-k_a^2)-2 \tau_0(k_v+h_wk_p) \ge 0, $ can be described by
		$${\mathcal S}_2:= \{k_p, k_v): k_p >0, \; k_v >0, \; \frac{k_v}{a_2} + \frac{k_p}{b_2} \le 1 \}, $$
	where 
		$$a_2:= \frac{1-k_a^2}{2 \tau_0}, \quad b_2:=\frac{(1-k_a^2)(1+k_a)}{4 \tau_0^2 (1+\eta)}. $$
	Clearly ${\mathcal S}_1, {\mathcal S}_2 \ne \emptyset$. To ensure that ${\mathcal S}_1 \cap {\mathcal S}_2 \ne \emptyset$,
	we need either $a_1 \le a_2$ or $b_1 \le b_2$. Considering the expressions for $b_1$ and $b_2$, we note that $a_1$ is always less than $a_2$ for any $\eta>0$. Hence, ${\mathcal S}_1 \cap {\mathcal S}_2 \ne \emptyset$ and one can find $(k_p, k_v) \in {\mathcal S}_1 \cap {\mathcal S}_2$.
	\end{itemize}
\end{proof}
\begin{IEEEbiography}{Dr. Swaroop Darbha}
received a Ph. D. in Mechanical Engineering from the University of California at Berkeley in 1994. He is currently a Professor of Mechanical Engineering at Texas A \& M University in College Station, TX.  His research interests include Advanced Vehicular Control and Diagnostic Systems, Motion planning and control of Unmanned Vehicles, Decision making under uncertainty, Fixed Structure Controller Synthesis.
\end{IEEEbiography}
\begin{IEEEbiography}{Shyamprasad Konduri}
receiveed his M.S in Mechanical Engineering from Oklahoma State Universty at Stillwater, OK in 2012. He is currently a doctoral student in the dept. of Mechanical Engineering at Texas A \& M University in College Station, TX. His research is focused on connected and autonomous vehicle systems, and mobile robotics. 
\end{IEEEbiography}
\begin{IEEEbiography}{Dr. Prabhakar R. Pagilla}
received a Ph. D. in Mechanical Engineering from the University of California at Berkeley in 1996. He is currently the James J. Cain Professor II in Mechanical Engineering at Texas A \& M University in College Station, TX. His research interests include modeling and control related problems in robotics/mechatronics, roll-to-roll manufacturing systems, autonomous vehicles, and large-scale nonlinear dynamic systems.
\end{IEEEbiography}

\end{document}